\documentclass[preprint,aps,pre]{revtex4-1}

\usepackage{latexsym,enumerate,amsmath,amssymb}
\usepackage{amsthm}

\newtheorem{theorem}{Theorem}
\newtheorem{lemma}{Lemma}

\newtheorem{corollary}[theorem]{Corollary}

\def\eg{{\it e.g.}}

\usepackage{graphicx}

\newcommand{\wt}[1]{\widetilde{#1}}

\begin{document}
\preprint{APS/123-QED}

\title{The Mori-Zwanzig formalism for the derivation of a fluctuating heat conduction model from molecular dynamics}

\author{Weiqi Chu}
\email{wzc122@psu.edu}
\affiliation{Department of Mathematics, the Pennsylvania State University, University Park, PA 16802-6400, USA.}%

\author{ Xiantao Li}
\email{xli@math.psu.edu}
\affiliation{Department of Mathematics, the Pennsylvania State University, University Park, PA 16802-6400, USA.}%

\date{\today}

\begin{abstract}
Energy transport equations are derived directly from full molecular dynamics models as coarse-grained description. With the local energy chosen as the coarse-grained variables, we apply the Mori-Zwanzig formalism to derive a reduced model, in the form of a generalized Langevin equation.  A Markovian embedding technique is then introduced to eliminate the history dependence. In sharp contrast to conventional energy transport models, this derivation yields  {\it stochastic} dynamics models for the spatially averaged energy. We discuss the approximation of the random force using both additive and multiplicative noises, to ensure the correct statistics of the solution. 
\end{abstract}


\maketitle

\section{Introduction}
During the past two decades, there has been rapidly growing interest in modeling heat transport at the microscopic scale. Such renewed interest has been driven by the progress in designing and manufacturing micro mechanical and electrical devices, for which thermal conduction properties have great influences on the performance and reliability. As the size of electrical and mechanical devices is decreased to the micron and sub-micron scales, they often exhibit heat conduction properties that are quite different from the observations familiar at macroscopic level. 
For example, extraordinarily large heat conductivity for carbon nano-tubes has been reported by many groups \cite{Blandin-2011,KiShMaMc-2001,VoZh-2012}, the conductivity of two-dimensional systems shows a strong dependence on the system size \cite{WaHuLi12}, which indicates the failure of the  Fourier's law, and heat pulses are observed in experiments \cite{tzou1995experimental}, which are typical behavior of wave equations, etc.   

From a modeling viewpoint, a natural approach to incorporate some of the observed behavior is to modify the traditional heat equation, \eg, by introducing nonlocal terms \cite{tzou1995experimental,Tzou-2001}. This approach is simple, but quite ad hoc. On the other hand, one may rely on the more fundamental description of phonons, which are the energy carriers in solids. The distribution of phonons is governed by the Peierles-Boltzmann (PB) equation \cite{Callaway59,Peierles55,Ziman01}, or simplified PB equations \cite{Chen01,JoMa93,Majumdar93,SvJuGo01,XuLi12}, where the cubic collision term is replaced by a relaxation term, similar to the Bhatnagar--Gross--Krook model \cite{bhatnagar1954model} in kinetic theory. In general, however, the computation of the full PB model is still an open challenge due to the high dimensionality.

Thus far, the most popular approach to study heat conduction is direct molecular dynamics (MD) simulations, which often mimics the experimental setup. Given an interatomic potential $V$, either  empirically constructed or derived from more fundamental considerations, a MD model is typically expressed in terms of the Newton's equations of motion.
There are many well established techniques for MD simulations \cite{AlTi89,FrSm02}. In particular, many interatomic potentials have been developed, which have shown great promises in predicting defect structures and phase configurations. Of particular interest to the present paper is that the phonon spectrum computed from some of the models is in good agreement with results from experiments or first-principle calculations (\eg, see \cite{SiPaSa97,WiMiHa2006}).   

In spite of the many constributions that have recently appeared to report the studies of  heat conduction processes (\eg, \cite{gill2006rapid,jolley2009modelling,Che2000,ChZhLi10,DoGa09,HeCh08,McKa06,SchPhKe02,VoCh00,Wang07,WaHuLi12,McKa2003,PoMa-2006,VoZh-2012}),
direct MD models have several serious limitations when applied to heat conduction problems. The first obvious limitation is the computational cost. Typical quantities of interest are expressed as ensemble averages or two-point correlations. For a non-equilibrium process such as the transient heat conduction process, the ensemble averages may not be replaced by time averages, at least very little theory exists to support such a practice. Therefore, many copies, typically tens of thousands, need to be created to average out statistical fluctuations. In addition, the size of the system (and time scale) that can be modeled by direct MD simulations is small, often  comparable or smaller than the mean free path of phonons. Most current MD studies are restricted to quasi-one-dimensional systems, \eg, nanowires \cite{DoGa09,li2003thermal,volz1999molecular,dames2004theoretical,wang2009thermal,donadio2009atomistic,lu2002size,chen2005effect,chen2004molecular}, nanotubes \cite{ChOk-2008,Fujii-etal-2005,KiShMaMc-2001,PoMa-2006,VoZh-2012,hone1999thermal,maruyama2002molecular,yu2005thermal,maruyama2003molecular,osman2001temperature}, and nanoribbons \cite{hu2009thermal,savin2010suppression,evans2010thermal}. They have motivated a lot of recent effort to understand the origin and limitations of Fourier's Law \cite{lepri1997heat,lepri2003thermal,lepri1998anomalous,bonetto2000fourier,garrido2001simple,bernardin2005fourier}.
Furthermore, most of the studies have been focused on the dependence of the heat conductivity on the geometry, length and temperature of the system. A typical setup is to connect the boundaries to two heat bath with different temperature, modeled by stochastic (Langevin) or deterministic (Nos\`e-Hoover \cite{nose1984molecular}) thermostats. 
The MD equations are solved to drive the system to a steady state, at which point the heat flux can be measured to estimate the heat conductivity. More general transient heat conduction problems, however, would require more substantial effort. 

Secondly, it is often straightforward to incorporate quantities such as displacement, velocity, temperature and pressure into MD simulations as constraints. However, temperature gradient is very difficult to impose. The temperature gradient that can be
imposed is  usually on the order of $10^9 - 10^8 K/m$, which is too large to model realistic systems. It is unclear whether results obtained this way can be appropriately extrapolated to the correct regime. Another way to interpret this is that a realistic temperature gradient, when applied to molecular systems, is too small to be incorporated accurately by the numerical methods. One is often confronted with the issue of small signal-to-noise ratio, a problem that also arises in fluid mechanics problems \cite{evans1984nonlinear,EvMo08,morriss1987application}.


This paper is strongly motivated by the above-mentioned issues, and the purpose  is to present a coarse-grained  (CG) model to alleviate these  fundamental modeling difficulties. The CG procedure drastically reduces the number of degrees of freedom and offers a practical alternative. Coarse-graining methodologies have found enormous application in material science problems and biological problems \cite{RuBr05,RuBr98,SiLe07,ChSt05,IzVo06,KaMaVl03,CoHoSh90,Li2009c,baaden_coarse-grain_2013,gohlke_natural_2006,golubkov_generalized_2006,gramada_coarse-graining_2011,nielsen_recent_2010,noid_perspective:_2013,noid_multiscale_2008,poulain_insights_2008,praprotnik_multiscale_2008,riniker_developing_2012,rudzinski_role_2012,shi_coarse-graining_2008,stepanova_dynamics_2007,zhang_systematic_2008}. Many CG  models have been developed and they have shown great promise in reducing the computational cost and efficiently capturing the primary quantities of interest.  However,  almost all the existing CG molecular models are focused on finding the effective potentials, known as the potential of mean forces,  at a {\it constant} temperature.  These existing CG models are in the similar form as the MD models with possible addition of damping terms or random forces. They typically describe only the time evolution of the averaged position and momentum.  

The current approach works with the local energy and aims at the energy transport process.
Starting with a locally averaged internal energy as CG variables, we use the Mori-Zwanzig formalism  \cite{Mori65,Zwanzig73} to first derive an {\it exact} equation for these variables. In particular, we choose Mori's orthogonal projection  \cite{Mori65} to project the equations to the subspace spanned by the CG variables. For such CG variables, this projection yields a memory term, which exhibits a simple form of a convolution in time. To alleviate the effort to compute the memory term at every step, we use the Markovian embedding techniques, recently developed in \cite{ceriotti2010colored,lei2016data,ma2016derivation}, to approximate the memory using an extended system of differential equations with {\it no} memory.  The coefficients for these approximations can be determined based on the statistics of the CG variables. 

In principle, the noise term can be averaged out by simply taking the average of every term in the CG model. This is a particular advantage of the Mori's projection \cite{ChSt05}. However,   motivated by the crucial observation that many mechanical systems at the micron scale or smaller are subject to strong fluctuations, we will forgo such an averaging step, and work with the models {\it with}  the random noise. This results in an energy transport model {\it with fluctuation}, represented by a system of stochastic differential equations (SDE).  Consequently, the solutions are expected to be stochastic in nature. An important issue naturally arises:  How does one guarantee that the corresponding solution has the correct statistics? We first consider an approximation of the random force by an additive white noise, in which case the solution should have Gaussian statistics. Unfortunately,  by examining a one-dimensional chain model, we have found that the correct statistics behaves more like a Gamma distribution. In particular, the energy must have a lower bound.  Although the approximation by additive noise yields reasonable approximations to the time correlations, the probability density function (PDF) of the solution is incorrect.

To ensure that correct PDF is obtained, we propose to approximate the noise by a {\it multiplicative} noise. In this case, the diffusion constant depends on the solution itself. We determine the diffusion coefficient by solving the steady-state Fokker-Planck equation. We are able to find a diagonal matrix for the diffusion coefficient matrix such that the Gamma distribution is an equilibrium probability density. As a further extension, we introduce a higher order approximation where both the CG energy variables and their time derivatives have the correct PDF. This leads to a Langevin type of equation with multiplicative noise.

We point out that one existing stochastic heat conduction model has been proposed by Ripoll et al in \cite{ripoll1998dissipative}, as an extension of the dissipative particle dynamics (DPD) \cite{espanol1995statistical}. The model was postulated, rather than derived from MD.

The rest of the paper is organized as follows: In section \ref{math}, we discuss the mathematical derivation, and examine the general properties of the generalized Langevin equations derived from the Mori's projection.  We introduce the Markovian embedding technique for the approximation of the memory term. Then in section \ref{1d}, we present the approximation of the random noise using a one-dimensional system  as a example.

\section{Mathematical derivation}\label{math}
\subsection{The general projection formalism}
Our starting point is an all-atom description, which embodies the detailed interactions among all the atoms in the system. More specifically, let $x$ and $v$ be displacements and velocities of the atoms respectively;  $x, v \in \mathbb{R}^{dN}$ with $d$ being the space dimension and  $N$ being the total number of atoms. The dynamics follows the Newton's second law,
\begin{equation}
\left\{
\begin{aligned}
\dot{x} &= v, \quad & x(0)=x_{0},\\
m\dot{v} &= f(x)=-\frac{\partial V(x)}{\partial x},\quad & v(0)=v_{0}, \\
\end{aligned} \right.
\label{eq: newton2}
\end{equation}
where $V(x)$ is the potential energy of the system. Solutions $\{x(t),v(t)| t\ge 0\}$ can be viewed as trajectories in the phase space  with an ensemble of initial states $(x_0,v_0).$ Following the notations in \cite{EvMo08}, we use $\Gamma=\mathbb{R}^{2dN}$ for the phase space and define the propagating operator $\mathcal{L}$ on $\Gamma$ as,
\begin{equation}
\mathcal{L} := v_{0}\cdot\frac{\partial } {\partial x_{0}} +\frac{f(x_{0})}{m}\cdot \frac{\partial}{\partial v_{0}}.
\end{equation}


We define a Hilbert space $\cal{H}$ equipped with an  inner product weighted by a probability density $\rho_0$. This corresponds to the initial preparation of the system. For any two $n$-dimensional  functions, $F, G: \Gamma \to \mathbb{R}^n,$ we define the average and the correlation matrix as follows,
\begin{equation}\begin{aligned}
 &\left<F\right>_{i} \overset{\Delta}{=} \int_{\Gamma} F_i(z) \rho_0(z) dz, \quad 1 \le i \le n, \\
 &\left<F, G\right>_{ij} \overset{\Delta}{=} \int_{\Gamma} F_i(z) G_j(z) \rho_0(z) dz, \quad 1 \le i, j \le n.
 \label{eq: minnerproduct}
\end{aligned}\end{equation} 

Suppose ${a}: \Gamma \to \mathbb{R}^{n}$ is a quantity of interest (QOI) and depends only on phase space variables $x$ and $v$ explicitly. For convenience, we work with a somewhat abused notation,
\begin{equation}
{a}(t) \overset{\Delta}{=} {a}\big(x(t),v(t)\big) \text{ and } {a} \overset{\Delta}{=} {a}(0).
\end{equation}
For the last part of our definition, we have followed the convention in  \cite{EvMo08}.

Our goal is to derive a reduced equation for the QOI $a(t)$, also known as coarse-grain (CG) variables. For this purpose, we follow the Mori-Zwanzig (MZ) procedure \cite{ChHaKu02,ChKaKu98,ChSt05,Mori65,Zwanzig73}. A key step in the MZ formulation is  a projection operator $\mathcal{P}$ that maps functions of $\Gamma$ to those of $a$.  We adopt the orthogonal projection suggested by Mori \cite{Mori65}. More specifically, for any function $G:\Gamma\to\mathbb{R}^{n}$, we define,
\begin{equation}
\mathcal{P}G\overset{\Delta}{=}  \left< G,a \right>M^{-1}a,
\end{equation}
where $M^{-1}$ is the inverse of matrix $M=\left< a,a\right>$, and the inner product of $G$ and $a$ is defined according to \eqref{eq: minnerproduct}.
We also define $\mathcal{Q}$ as the complementary operator of $\mathcal{P}$, i.e. $\mathcal{Q}=\mathcal{I}-\mathcal{P}$. Note that the covariance matrix $M$ only involves the one-point statistics of $a$ and can be guaranteed to be non-singular by carefully selecting the CG variables. In practice, this corresponds to the appropriate choice of $a$ so that the CG variables are not redundant. Even in the case when  the CG variables are redundant, e.g., when there is energy conservation and the matrix $M$ becomes singular, the projection can still be well defined by interpreting $M^{-1}$ as the pseudo-inverse.

Once the projection operator is in place, the Mori-Zwanzig formalism can be invoked, and the following generalized Langevin equation (GLE)  can be derived \cite{Mori65},
\begin{equation}
\dot{a}(t) = \Omega a(t) - \int_{0}^{t} \theta(t-s)a(s) ds + F(t),
\label{eq: GLE}
\end{equation}
where 
\begin{equation}\label{mz}
\Omega=\left<\mathcal{L}a,a\right>M^{-1},\; F(t) = e^{t\mathcal{QL}}\mathcal{QL}a, {\text{ and }}\; \theta(t)=-\left<\mathcal{L}F(t),a\right>M^{-1}. 
\end{equation}

For the choice of $\rho_0,$ we pick an equilibrium probability density. More specifically, let $\mathcal{L}^{*}$ be the $L^{2}$-adjoint operator of $\mathcal{L}$. Then for any  equilibrium density $\rho_{eq}$ that satisfies $\mathcal{L}^{*}\rho_{eq}=0$ we can set $\rho_0=\rho_{eq}.$  When the system is near equilibrium, this serves as the first approximation. Further corrections can be made using the linear response approach \cite{Toda-Kubo-2}. For a Hamiltonian dynamics like \eqref{eq: newton2}, we have $\mathcal{L}^{*}=-\mathcal{L}$  \cite{EvMo08}. We pick the canonical ensemble for $\rho_{eq},$
\begin{equation}
 \rho_{eq}= \frac{1}{Z} e^{-\beta H}.
\end{equation}
In principle, other forms of the probability density, especially those obtained from the maximum entropy principle \cite{Zwanzigbook,Zwan80}, can be used as well.

Several properties can be deduced from the derivation. They are summarized in the following theorem.

\begin{theorem}
Assuming that $\left<a\right>=0$, then the following properties hold,
\begin{equation}\label{eq: prop-F}
\begin{aligned}
\left<F(t)\right> &=0, \qquad \forall \; t\ge 0,\\
\theta(t_1 -t_2) &=  \left<F(t_1),F(t_2) \right>M^{-1}, \quad \forall \; t_{1},t_{2}\ge 0 \text{ and } t_{1}\ge t_{2},\\
\left<F(t),a\right> &=0, \qquad \forall \; t\ge 0.\\
\end{aligned}
\end{equation}
\end{theorem}
\begin{proof}
Note that with the inner product defined above, the adjoint operator of $\mathcal{L}$ is $-\mathcal{L}$, and $\mathcal{P}$ and $\mathcal{Q}$ are self-adjoint, i.e., $\left< \mathcal{L}b,c\right>=-\left< b,\mathcal{L}c\right>$ and $\left< \mathcal{P}b,c\right>=\left< b,\mathcal{P}c\right>$ for any $b,c: \Gamma \to \mathbb{R}^n$. 

Now for the first property, we proceed as follows,
\begin{equation}\begin{aligned}
 \left<F(t)\right> &=  \left<\mathcal{QL}e^{t\mathcal{QL}}a\right> =  \left< \mathcal{L}e^{t\mathcal{QL}}a\right> - \left<\mathcal{PL}e^{t\mathcal{QL}}a\right> \\ &= -\left< e^{t\mathcal{QL}}a\mathcal{L}1\right> - \left<\mathcal{L}e^{t\mathcal{QL}}a,a\right>M^{-1}\left< a\right> = 0.
\end{aligned}\end{equation}

For the second property, we start with the second equation in \eqref{mz} and we get, 
\begin{equation}\begin{aligned}
\theta(t_{1}-t_{2}) &=-\left<\mathcal{L}F(t_{1}-t_{2}),a\right>M^{-1} = \left<e^{(t_{1}-t_{2})\mathcal{QL}}\mathcal{QL}a,\mathcal{L}a\right>M^{-1}\\
&=\left<F(t_{1}),e^{t_{2}\mathcal{LQ}}\mathcal{L}a\right>M^{-1}=\left<F(t_{1}),\mathcal{Q}e^{t_{2}\mathcal{LQ}}\mathcal{L}a\right>M^{-1} \\
&= \left<F(t_1),F(t_2) \right>M^{-1}.
\end{aligned}\end{equation} 

Finally, since $\mathcal{P}a =a$ and $\mathcal{Q}F(t)=F(t)$, one can easily verify that,
\begin{equation}
\left< F(t),a\right> = \left< \mathcal{Q}F(t),\mathcal{P}a\right> = \left< F(t),\mathcal{QP}a\right> = 0.
\end{equation}
 \end{proof}

The first two equations imply that the random process, $F(t),$ is a stationary random process in the wide sense \cite{chorin2009stochastic}, and it also satisfies the second fluctuation-dissipation theorem (FDT). It is a necessary condition for the solution to have the correct variance \cite{Kubo66}. The last condition suggests that the random force and the initial value of $a$ are uncorrelated. The last two properties have also been discussed  in the Mori's original paper \cite{mori1965b}. 

Using the first property, one can take an average of the GLE \eqref{eq: GLE}, and arrive at a deterministic system,
\begin{equation}
\dot{\left<a\right>}(t) = \Omega \left<a\right>(t) - \int_{0}^{t} \theta(t-s)\left<a\right>(s) ds,
\label{eq: GLE-avg}
\end{equation}
which is a set of integral-differential equations describing the time evolution of the average of the quantity of interest \cite{ChKaKu98}. This is often seen as a particular advantage of Mori's projection. However, in this paper, we are more concerned with the quantity $a$ with fluctuation. 

\subsection{An explicit representation of the random noise}
In the general MZ formalism, the random noise has been expressed in a quite abstract form. The practical implementation is rather difficult in general. Here, we provide a more detailed characterization of the random noise, by embedding it in  an infinite system of ordinary differential equations. Important properties can also be deducted from these differential equations. 

Suppose that $\{\mathcal{L}^{j}a\}_{j\ge 0}$ are linearly independent, by inspecting the first few terms in the exponential operator, one finds that the random force term can be written as,
\begin{equation}\label{eq: exp-F}
F(t) = \sum_{j\geq 0} C_{j}(t)\mathcal{L}^{j}a.
\end{equation}
Together with the orthogonal dynamics, $\dot{F}(t)=\mathcal{QL}F(t)$, we can derive a set of equations for the coefficients,
\begin{equation}
\begin{aligned}
\dot{C}_j(t) = &C_{j-1}(t), \quad j \ge 1, \\
\dot{C}_0(t) =&- \sum_{j\ge0} C_j (t)M_{j+1}M^{-1},
\end{aligned}
\label{eq: odeC}
\end{equation}
where $M_{j}$ are referred to as the {\it moments} associated with the statistics of $a$, defined as follows,
\begin{equation}
 M_j \overset{\Delta}{=} \left<\mathcal{L}^ja,a\right>=\left<
 \frac{d^j}{dt^j}a,a\right>.
\end{equation}
The initial condition is given by, 
\begin{equation}
C_{0}(0)=-M_{1}M^{-1}, \; C_{1}(0)=I \text{ and } C_j(0)=0, \; \forall j\ge 2.
\end{equation}
Therefore, the random noise can be characterized via an infinite set of ordinary differential equations.

\subsection{Properties of the kernel function}
Thanks to the explicit representation of the random force and the FDT \eqref{eq: prop-F}, certain values of the memory kernel can be reconstructed or approximated using the equilibrium properties. In this section, we will explain how the connections can be made.

Direct calculations yield,
\begin{equation}
\theta(t) = -\sum_{j\ge 0}C_{j}(t) M_{j+1}M^{-1},
\end{equation}
where $C_{j}(t)$ are the coefficients of $F(t)$ in the expansion \eqref{eq: exp-F} and are given by the solution of \eqref{eq: odeC}. The derivatives of $\theta(t)$ at $t=0$ can be written out explicitly, as shown by the following theorem, which can be proved by direct substitutions.

\begin{theorem}
$C_{0}^{(k)}(0)$ can be computed recursively, and they satisfy,
\begin{equation}
C_{0}^{(k)}(0) = -\sum_{j=-1}^{k-1}C_{0}^{(j)}(0)\overline{M}_{k-j-1}, \;\; k\ge 0.
\end{equation}
Furthermore, the derivatives of the memory kernel are given by,
\begin{equation}\begin{aligned}
\theta^{(k)}(0) = -\sum_{j=-1}^{k} C_{0}^{(j)}(0)\overline{M}_{k-j+1}, \;\; k\ge 0,
\label{eq: dtheta}
\end{aligned}\end{equation}
where  we defined $C_{0}^{(-1)}\overset{\Delta}{=}C_{1}$ and $\overline{M}_{j}\overset{\Delta}={M}_{j}M^{-1}$.

\end{theorem}

Based on the theorem, one can write down the first few derivatives of $\theta(0)$, 
\begin{equation}
\begin{aligned}
  \theta(0)=&  -\overline{M}_2+\overline{M}_{1}^{2} ,\\
  \theta'(0)=& -\overline{M}_3+\overline{M}_{2}\overline{M}_{1}+\overline{M}_{1}\overline{M}_{2}-\overline{M}_{1}^{3},\\
  \theta''(0)=& -\overline{M}_{4}+\overline{M}_{3}\overline{M}_{1}+\overline{M}_{2}^{2}+\overline{M}_{1}\overline{M}_{3}-\overline{M}_{2}\overline{M}_{1}^{2}-\overline{M}_{1}\overline{M}_{2}\overline{M}_{1}-\overline{M}_{1}^{2}\overline{M}_{2}+\overline{M}_{1}^{4}, \\
   \cdots  \\
\end{aligned}
\end{equation}

This routine allows us to express the values of the kernel function at $t=0$ in terms of the {\it equilibrium} statistics of the CG variables. The approximation scheme in the next section takes advantage of these properties. 

\subsection{Markovian embedding -- a systematic approximation of the memory term}

A well known practical issue associated with the solution of the GLE is that computation of the memory term. Clearly, a direct evaluation of the integral requires the storage of the solutions from all previous steps, and such evaluations have to be carried out at every time step. To alleviate such effort, we will use the Markovian embedded technique and approximate the memory term via an extended system of equations \cite{lei2016data}. The idea is to incorporate the aforementioned values of the kernel function  into the Laplace transform of $\theta$. More specifically, we define, 
\begin{equation}
 \Theta(\lambda)= \int_0^{+\infty} \theta(t) e^{-t/\lambda} dt.
\end{equation}

As $\lambda \to 0+,$ using integration by parts repeatedly, we find that 
\begin{equation}
 \Theta(\lambda) = \lambda \theta(0) +  \lambda^2 \theta'(0) +  \lambda^3 \theta''(0)  + \lambda^4 \theta'''(0) + \cdots.
\end{equation}

We now turn to the limit as $\lambda \to +\infty$, which embodies long-time behavior of the kernel function. For this calculation, we start with the GLE \eqref{eq: GLE}, multiply both side by $a^{\intercal},$ and take the average. Let $M(t)=\left<a(t),a\right>,$ then we have,
\begin{equation}
  \dot{M}(t)= \Omega M(t) - \int_0^t \theta(t-s) M(s)ds.
  \label{eq: odeM}
\end{equation}

Let $\wt{M}(\lambda)$ be the Laplace transform of $M(t)$. Taking the Laplace transform of \eqref{eq: odeM}, we find,
\begin{equation}
 \frac{1}{\lambda} \wt{M}(\lambda) - M= \Omega \wt{M}(\lambda) - \Theta(\lambda) \wt{M}(\lambda),
\end{equation}
which yields,
\begin{equation}\label{eq: th-inf}
 \Theta(+\infty)= \Omega + M \Big[ \int_0^{+\infty} M(t) dt\Big]^{-1}.
\end{equation}

Again this is related to the statistics of $a$. We now incorporate the values of $\Theta$ from both regimes: $\lambda \to 0+$ and $\lambda \to +\infty.$ Such two-sided approximations, which are similar to the Hermite interpolation problems, have demonstrated promising accuracy over both short and long time scales \cite{lei2016data}.


The idea of the Markovian embedding is to  approximate the memory term by rational functions in terms of the Laplace transform. In general, we can consider a rational function in the following form,
\begin{equation}
 R_{k,k}= \big[I-\lambda B_1 - \cdots - \lambda^k B_{k}\big]^{-1}\big[A_{0}+\lambda A_1 + \lambda^2 A_2 + \cdots + \lambda^k A_{k}\big]. 
\end{equation}
The coefficients in the rational function can be determined based on the values of the kernel functions, e.g., those presented in the previous section.

When $k=0,$ we are led to a constant function, $R_{0,0}=\Gamma$, which, we choose to be given by \eqref{eq: th-inf}:
\begin{equation}
 \Gamma=  \Theta(+\infty).
\end{equation}

In the time domain, this amounts to approximating the kernel function by a delta function:
\begin{equation}
   \int_{0}^{t} \theta(t-s)a(s) ds \approx \Gamma a(t).
\end{equation}
This is often referred to as the Markovian approximation \cite{hijon2006markovian,kauzlaric2012markovian}.  

When $k=1,$ we have,
\begin{equation}
 R_{1,1}(\lambda) = \big[I - \lambda B_1]^{-1} \big[ A_0+\lambda A_1].
\end{equation}

To determine the coefficients, we match the following values,
\begin{equation}\label{eq: R11}
R_{1,1}(0)= \Theta(0), \; R_{1,1}'(0)= \Theta'(0)\; \text{ and }  R_{1,1}(+\infty)= \Theta(+\infty),
\end{equation}
which yield,
\begin{equation}
A_{0}=0,\; A_1= \theta(0) \; \text{ and } B_1= -A_1\Theta(+\infty)^{-1}.
\label{eq: A1B1}
\end{equation}
In the time-domain, this corresponds to an approximation of the kernel function by an matrix exponential $e^{B_{1}t}A_{1}$. 

Meanwhile, if we define the memory term as $z,$ 
\begin{equation}\label{eq: z}
z= \int_{0}^{t} \theta(t-s)a(s) ds,
\end{equation}
we can write down an auxiliary equation,
\begin{equation}
 \dot{z}= A_1 a + B_1 z. 
\end{equation}
This way, the memory term is embedded in an extended dynamical system {\it without} memory.

As the order of the approximation $k$ increases, we obtain an hierarchy of approximations for the memory term, which can be written as a larger extended system of equations \cite{lei2016data,ma2016derivation}.  We will not discuss the higher order approximations in this paper.

\bigskip

It remains to approximate the random noise term. This will be discussed in the next section, along with a specific example of the MD model. 

\section{Application to energy transport} \label{1d}
\subsection{A one-dimensional example}
Let's consider a one-dimensional isolated chain model of $N$ atoms and they are evenly divided into $n$ blocks, each of which contains $\ell$ atoms, as shown in Figure \ref{fig: 1dchain}; $N=n\ell $. We will focus on the study of energy transport between  these blocks.

\begin{figure}[htp]
\begin{center}
\includegraphics[scale=0.9]{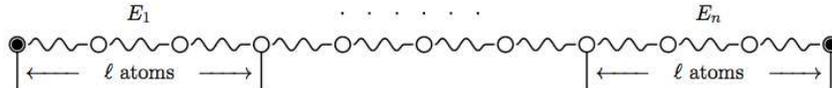}
\caption{1-D chain of particles. Every $\ell$ atoms are grouped into one block. \label{fig: 1dchain}}
\label{fig: const}
\end{center}

\end{figure}

Let $x$ and $v$ be the displacements and velocities respectively, satisfying \eqref{eq: newton2}. Periodic boundary conditions are imposed. Let $S_{I}$ is the index set of $I$-th block, labeled as, $S_{I}=\{\ell (I-1)+1,\cdots,\ell I\}$, $I=1,2,\cdots,n$. Let $\phi(x_{i}-x_{j})$ be the pairwise potential coming from interactions between the $i$th and $j$th atoms. Here, we use the Fermi-Pasta-Ulam (FPU) potential,
\begin{equation}
\phi(r) = \frac{r^{4}}{4}+\frac{r^{2}}{2}. 
\end{equation}
If we only consider the nearest neighbor interactions, the potential energy of this 1-d chain is given by, 
\begin{equation}
V(x) = \sum_{i=1}^{N}\phi(x_{i-1}-x_{i}).
\end{equation}

We define the local energy associated with the $I$-th block as follows,
\begin{equation}
E_{I}(t) = \sum_{i\in S_{I}} \frac{1}{2}mv_{i}^{2} + \frac12\phi(x_{i}-x_{i-1})+\frac12\phi(x_{i+1}-x_{i}).
\label{eq: edef}
\end{equation}

The rate of change of the local energy can be attributed to the energy flux $J$, 
\begin{equation}
\dot{E}_{I}(t) = J_{I+\frac12}(t) - J_{I-\frac12}(t),
\end{equation}
where $J_{I+\frac{1}{2}}$ is the energy flux between the $I$th and $(I+1)$th blocks. Direct calculation yields,
\begin{equation}
J_{I+\frac12} = \frac12 \phi'(x_{\ell I+1}-x_{\ell I}) (v_{\ell I+1} + v_{\ell I}).
\end{equation}
Notice that the energy flux only depends on the atoms next to the interfaces between two adjacent blocks. 

\subsection{The generalized Langevin equation for the energy}
As an application of the Mori's projection method, we define the quantity of interest $a$ as the shifted energy of blocks,
\begin{equation}
a_{I}(t) \overset{\Delta}{=} E_{I}(t) - \left< E_{I}(0)\right>.
\label{eq: adef}
\end{equation}
The subtraction of the average is to ensure that $\left<a\right>=0.$ 

\begin{theorem}
If $a(t)$ is as defined in \eqref{eq: edef} and \eqref{eq: adef}, the odd moments of $a$ vanish, i.e., 
\begin{equation} M_{2j+1}=0, \quad \forall \;j \geq 0. \end{equation}
\label{thm: modd}
\end{theorem}

The following lemma can be easily established, as a preparation to prove this theorem, 
\begin{lemma}
If $J(x,v)=G(x)F(v)$ is a scalar, separable function, then there exist $m\ge 1$, $G_{i}$ and $F_{i}$ s.t.
\begin{equation}
\dot{J}(x(t),v(t)) \overset{\Delta}{=} \frac{d}{dt} J(x(t),v(t)) = \sum_{i=1}^{m} G_{i}(x(t))F_{i}(v(t)).
\end{equation} 
Furthermore, if $F(v)$ is an odd function w.r.t.  $v$, then $\dot{J}(x,v)$ is an even function w.r.t $v$. If $F(v)$ is even w.r.t. $v$, then $\dot{J}$ is odd.
\label{lemma: gf}
\end{lemma}
Now let's turn to the proof of Theorem \ref{thm: modd}.
\begin{proof}
Write $ \overline{E} \overset{\Delta}{=} \langle E(0) \rangle$, and consider $j\geq 0$,
\begin{equation}\begin{aligned}
M_{2j+1} 	&= \langle \mathcal{L}^{2j+1}a,a \rangle \\
		&= \langle \mathcal{L}^{2j+1}E(0), E(0) -\overline{E} \rangle -  \langle \mathcal{L}^{2j+1}\overline{E}, E(0) -\overline{E} \rangle \\
		&= \langle \mathcal{L}^{2j+1}E(0), E(0) \rangle + \langle \mathcal{L}^{2j}E(0), \mathcal{L}\overline{E} \rangle \\
		&= \left< \frac{d^{2j+1}}{dt^{2j+1}}E(0), E(0) \right>.
\end{aligned}\end{equation}
By induction and Lemma \ref{lemma: gf}, one can verify that $\frac{d^{2j+1}}{dt^{2j+1}}E(0)$ is an odd function w.r.t $v$. $E(0)$ is even w.r.t $v$, so when we integrate the product over the velocity domain weighted by a Gaussian distribution, we get a zero average. 
\end{proof}

With this result, we can see that the Markovian term  in the GLE \eqref{eq: GLE} must be zero, i.e., $\Omega=0$. Furthermore, we have
\begin{corollary}
The derivatives of the memory function at $t=0$ are given by,
\begin{equation}\begin{aligned}
&C_{0}^{2j}(0) = 0, \quad j\geq 0, \\
&\theta^{(2j)}(0)= -\sum_{i=0}^{j}C_{0}^{(2i-1)}\overline{M}_{2(j+1-i)}, \quad j\geq 0,  \\
&\theta^{(2j+1)}(0) = 0, \quad j\geq 0.\\
\end{aligned}\end{equation}
\end{corollary}

For instance, the first few derivatives of $\theta(0)$ in even orders are listed below, 
\begin{equation}
\begin{aligned}
  &\theta(0)= -\overline{M}_2, \\
  &\theta''(0)= -\overline{M}_{4}+\overline{M}_{2}^{2}, \\
  &\theta^{(4)}(0) = -\overline{M}_{6}+\overline{M}_{4}\overline{M}_{2}+\overline{M}_{2}\overline{M}_{4}+\overline{M}_{2}^{3},\\
   & \cdots \cdots \\
\end{aligned}
\end{equation}

Before we discuss the approximation of the random noise, we first start with a full MD simulation, from which we can obtain the time series  of $a(t)$ and $F(t)$. Our numerical test simulates an equilibrium system containing $500$ atoms which are evenly divided into $n=50$ blocks. The histograms of one entry of $a(t)$ and $F(t)$ at equilibrium are shown in Figure \ref{fig: afpdf}, which can be assumed to be  the exact stationary distributions.  Interestingly, both quantities exhibit non-Gaussian statistics. The PDF of $a(t)$ fits perfectly to a shifted Gamma distribution, and the PDF of the random noise follows a Laplace distribution. See Figure \ref{fig: afpdf}.

\begin{figure}[ht]
\includegraphics[scale=0.53]{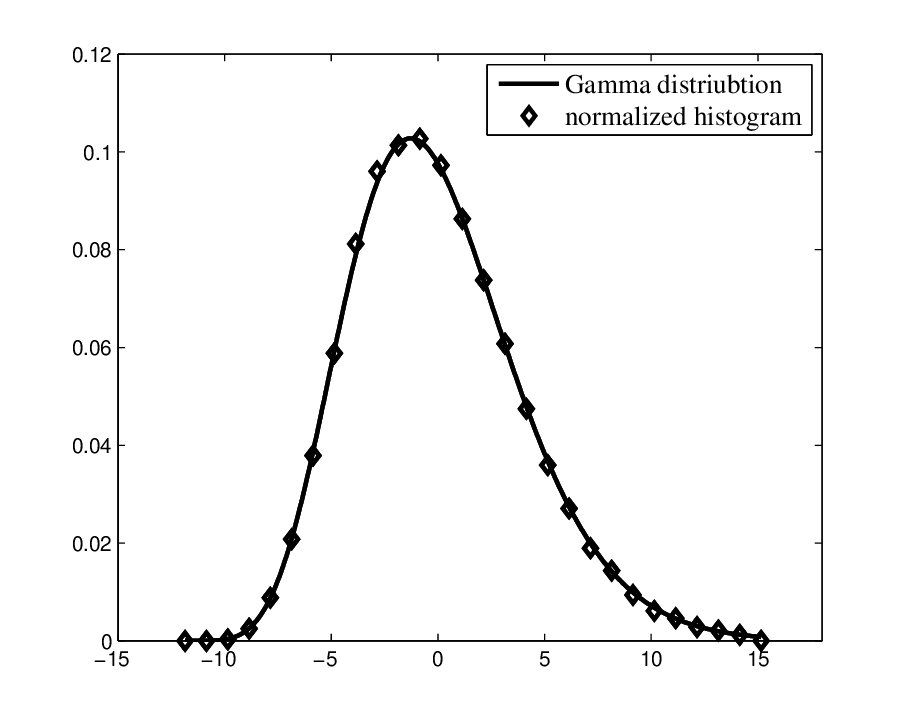}
\includegraphics[scale=0.53]{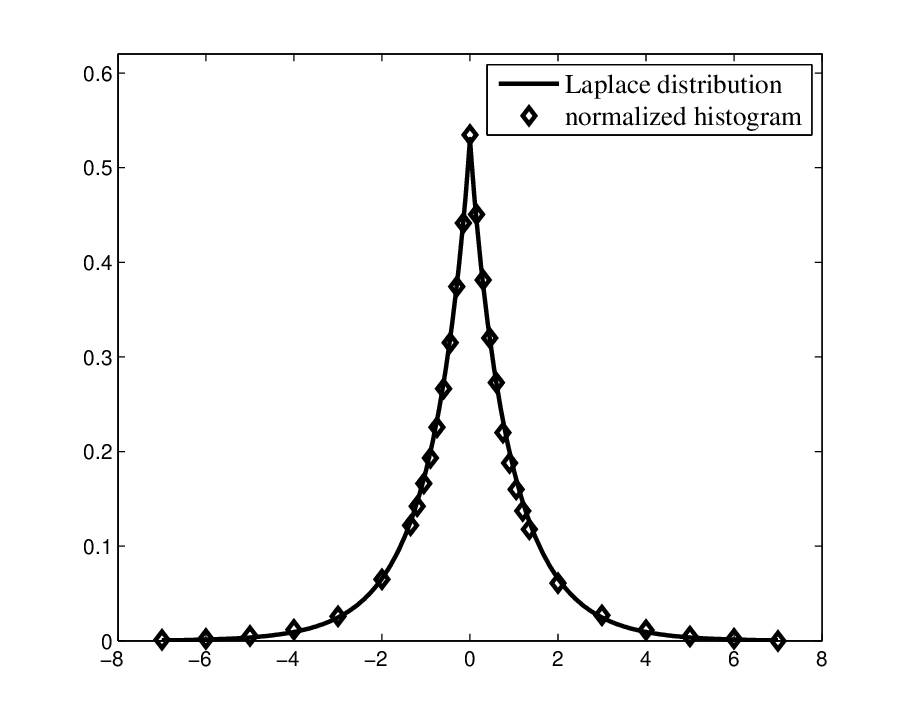}
\caption{The left figure shows the histogram of $a_{1}(t)$ obtained from direct simulations. The data fits well to a shifted Gamma distribution $\Gamma_{\alpha,\beta}(x) \sim (x+\mu)^{\alpha-1}\exp(-\beta(x+\mu)), \; x\ge -\mu,$ with $\alpha=10.0596, \; \beta=0.7825$ and $\mu=\alpha/\beta$. The right figure shows the histogram of $F_{1}(t)$ and it fits a Laplace distribution $\text{Lap}_{b}(x)\sim\exp(-b|x|)$ with $b=1.0801$. 
}
\label{fig: afpdf}
\end{figure}

In the following section, we will focus on the approximation of the noise term so that the solution of the reduced models gives consistent statistics of the CG energy variables.

\subsection{Approximation of the noise}

In the previous section, we have discussed how the memory term can be approximated using the rational approximation in terms of the Laplace transforms. In particular, it gives rise to deterministic (or drift) terms in the resulting approximate models. Here, we discuss the approximation of the noise. We will consider both additive and multiplicative noises.

\subsubsection{Approximations by additive noise}

A natural (and most widely used) approximation is by a Gaussian white noise. For instance, for the first order approximation, we are led to a linear SDE,
\begin{equation}
   \dot{{a}}(t)  = -\Gamma {a}(t) + \sigma \zeta(t),
   \label{eq: adsde0}
\end{equation} 
where $\zeta(t)$ is the standard Gaussian-white noise,
\begin{equation}
  \left< \zeta(t_1), \zeta(t_2)^\intercal\right> = \delta(t_1-t_2). 
\end{equation}

In order for the solution $a$ to have the correct covariance $M$, the parameter $\sigma$ has to satisfy the Lyapunov equation, \cite{risken1989fpe}
\begin{equation}
\Sigma \overset{\Delta}{=} \sigma \sigma^\intercal = \Gamma M + M \Gamma^\intercal.
\end{equation}

On the other hand, with the rational  approximation \eqref{eq: R11} of the kernel function, we may introduce noise via the second equation. Namely,
\begin{equation}
\left\{
\begin{aligned}
  \dot{{a}}(t)  = & - z(t), \\
  \dot{z}(t)  = & B_1 z(t) + A_1 {a}(t) + \sigma\zeta(t).
\end{aligned}
\right.
\label{eq: adsde1}
\end{equation}

It is clear that the second equation can be solved explicitly and substituted into the first equation, which would yield a similar equation to the GLE \eqref{eq: GLE}. By choosing the initial condition $z(0)$ and $\Sigma$ appropriately, the approximations to the memory kernel and the random noise can be made consistent, in terms of the second FDT \eqref{eq: prop-F}.

\begin{theorem}
  Assuming the covariance of $z(0)$ is $A_1$, and
  \begin{equation}
    B_1 A_1 + A_1 B_1^\intercal + \Sigma=0,
  \end{equation}
  then, the extended system is equivalent to approximating the kernel function by $\theta_1(t)=e^{tB_1}A_1$, and the approximate noise, denoted by $F_1(t)$, to $F(t)$ satisfies the second FDT exactly. Namely, 
  \[\theta_1(t -t') =  \left<F_1(t),F_1(t') \right>M^{-1}, \quad \forall\; t,t'\ge0.\]
    
\end{theorem}
The proof of this theorem can be found in \cite{ma2016derivation}.

The approximation by additive noises inevitably leads to a Gaussian distribution for $a(t)$ \cite{risken1989fpe}. To check the validity of this assumption, we ran a direct simulation based on both the zeroth order model \eqref{eq: adsde0} and the first order model \eqref{eq: adsde1}.  
The solutions are then compared to the true statistics, and the results are displayed in  Figure \ref{fig: addcorr}. 
\begin{figure}[htp]
\begin{center}
\includegraphics[width=0.45\textwidth]{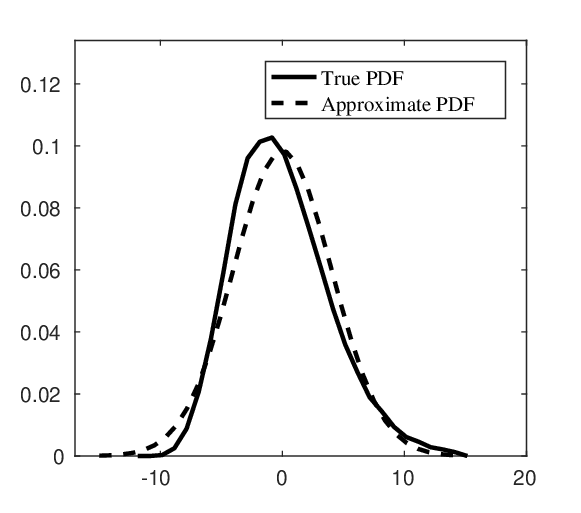}
\includegraphics[width=0.45\textwidth]{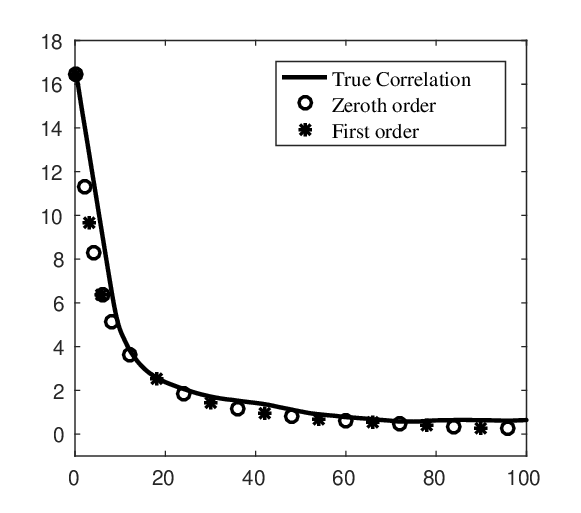}
\caption{The left figure shows the true PDF of $a_1$ and approximate PDF from approximations with Gaussian additive noise. The right panel shows the time correlation of approximate models.}
\label{fig: addcorr}
\end{center}
\end{figure}

From Figure \ref{fig: addcorr}, we observe that although the time correlation of the CG energy is well captured, the PDF deviates from the true distribution. 


\subsubsection{Approximations using multiplicative noise}
As alluded to in the previous section, the approximate model driven by Gaussian additive white noise may not capture the correct PDF. In this section, we consider multiplicative noise, with the objective of enforcing the correct equilibrium statistics for the solution of the SDEs. 
 
We start with a further observation that the energy of each block is almost independent to each other. Figure \ref{fig: pdfs} shows the agreement between the joint histogram of the energy of two adjacent blocks and the true marginal PDFs.  It is interesting that same observations have been made for biomolecules \cite{faure2017entropy}. It is clearly difficult to prove the exact independence theoretically. Therefore, we keep this as our main assumption, and postulate the stationary PDF ($\rho(a)$) of the energy as a shifted multi-Gamma distribution with parameters $\alpha_i$ and $\beta_i$,
\begin{equation}\label{multi-gamma}
\rho(a) \propto \prod_{i=1}^{n} \left({a}_{i}-\frac{\alpha_{i}}{\beta_{i}}\right)^{\alpha_{i}-1}\!\!\!\!\!e^{-\beta_{i}\left({a}_{i}+\alpha_{i}/\beta_{i}\right)}.
\end{equation}

\begin{figure}[htp]
\includegraphics[scale=0.6]{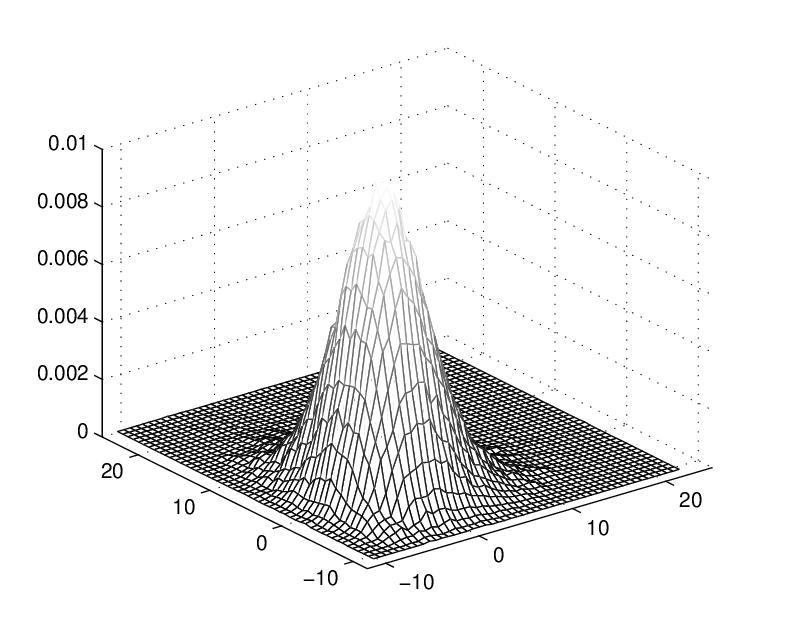}
\includegraphics[scale=0.5]{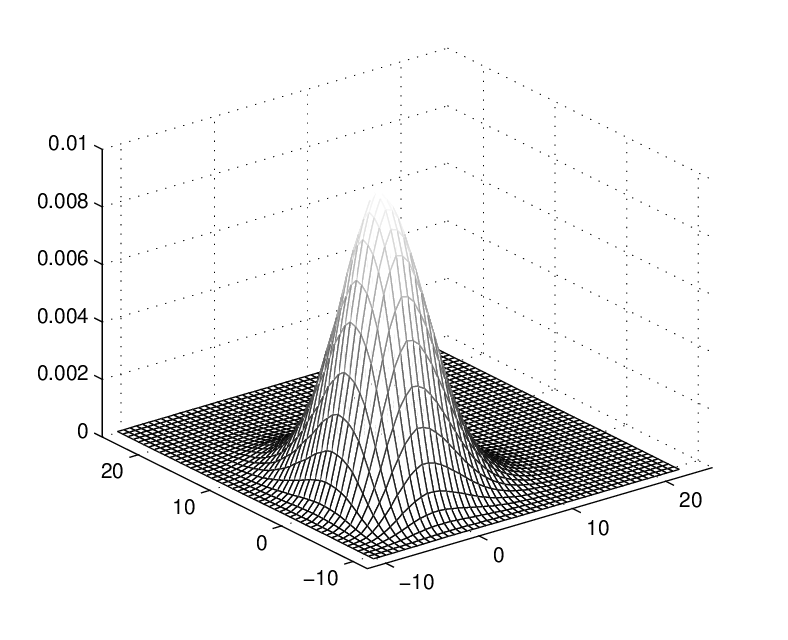}
\caption{The left figure shows the normalized joint histogram of $a_{1}$ and $a_{2}$ and the right panel shows the product of two normalized individual histogram of $a_{1}$ and $a_{2}$. 
}
\label{fig: pdfs}
\end{figure}

Now we reconsider the zeroth order approximation, the Markovian approximation by $R_{0,0}=\Theta(+\infty)=\Gamma.$ With a multiplicative noise, we are solving the following SDE,
\begin{equation}
 \dot{a}(t) = -\Gamma a(t) + \sigma(a(t)) \xi(t),
 \label{eq: osde}
\end{equation}
where $\xi(t)$ is the again standard Gaussian white noise. The SDE is interpreted in the It\^{o} sense. 

To derive a simple formula, we seek $\sigma$ in a diagonal form.  We aim to construct $D\overset{\Delta}{=}2\sigma\sigma^{\intercal}$ to ensure the desired PDF given by \eqref{multi-gamma}. By simplifying the Fokker-Planck equation (FPE) that corresponds to \eqref{eq: osde}, we obtain,
\begin{equation}
\frac{\partial D_{ii}\rho}{\partial a_{i}}= -\rho\left( \Gamma_{ii}a_{i}+\sum_{j=1,j\neq i}^{n} \Gamma_{ij}a_{j}\right).
\label{eq: Dii1}
\end{equation}

By directly solving these differential equations, we obtained an explicit formula for the matrix $D$, as summarized in the following theorem, which can be proved by direct integration of \eqref{eq: Dii1}. 

\begin{theorem}
If $\Gamma$ is a Z-matrix, i.e., the off-diagonal entries are non-positve, and $\Gamma$ is semi-positive definite, then there exists a diagonal matrix $D$ for which the multi-Gamma distribution  \eqref{multi-gamma} is a steady state solution of the Fokker-Planck equation. The diagonals of $D$ are given by,
\begin{equation}
D_{ii} = \frac{\Gamma_{ii}}{\beta_{i}}\left( a_{i}+\frac{\alpha_{i}}{\beta_{i}}\right) -\sum_{j=1, j\neq i}^{m}\Gamma_{ij} \frac{a_{j}\int_{-\frac{\alpha_{i}}{\beta_{i}}}^{a_{i}}\rho_{i}(x)dx+ \alpha_{j}/\beta_{j}}{\rho_{i}(a_{i})} \geq 0,
\end{equation}
where $\rho_{i}$ is the marginal PDF of $a_{i}$ and $\rho_{i}$ is propontional to $\left({a}_{i}+\frac{\alpha_{i}}{\beta_{i}}\right)^{\alpha_{i}-1}\!\!\!\!\!e^{-\beta_{i}\left({a}_{i}+\alpha_{i}/\beta_{i}\right)}$. 
\end{theorem}

As a verification, we solved the SDEs \eqref{eq: osde} with coefficients determined by the above formula.  Figure \ref{fig: multi} shows the PDF of the first component along with its time correlation. Due to the uniform partition of the system, we expect the statistics to be the same for all the components of $a$. So we will only show the properties of the first component $a_1.$ It is clear that they are both consistent with the truth.  It is worthwhile to point out that the SDE \eqref{eq: osde} contains a diffusion coefficient $\sigma$ which is unbounded. This can be viewed as a mechanism for the energy to stay above a lower bound. However, this introduces a stiff problem for the numerical approximation. To resolve this numerical issue, we applied the implicit Taylor method \cite{Tian2001implicit}.

\begin{figure}[htp]
\begin{center}
\includegraphics[scale=0.58]{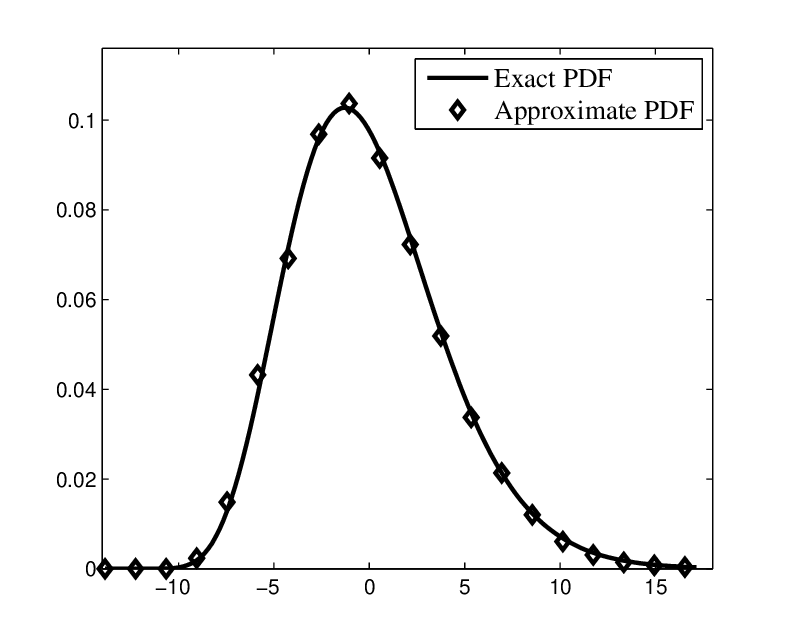}
\includegraphics[scale=0.58]{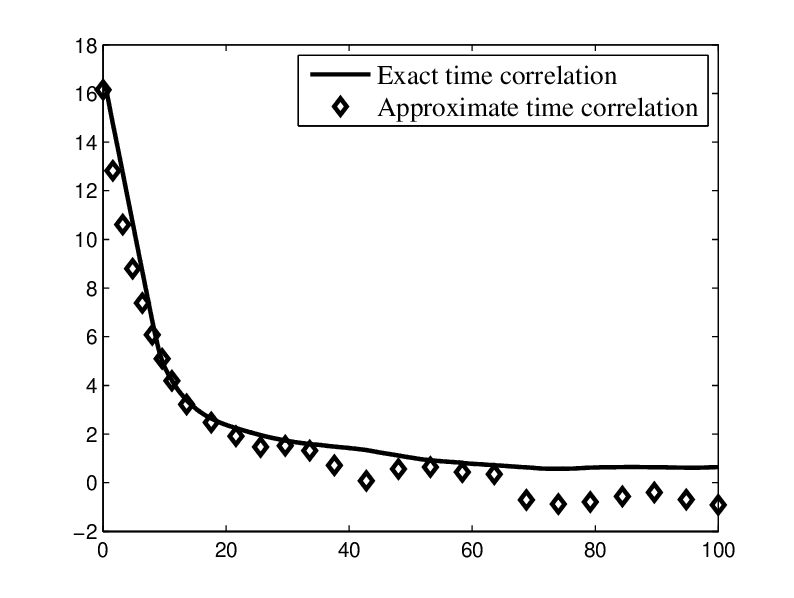}
\caption{The left figure shows the approximate PDF of $a_{1}$ from the multiplicative model \eqref{eq: osde}, along with the exact PDF. The right panel indicates the time correlation function.}
\label{fig: multi}
\end{center}
\end{figure}

Finally  let's turn to the model obtained by the first order approximation of the memory term. With Gaussian multiplicative noise, the first order model  can be written formally as follows,
\begin{equation}
\left\{
\begin{aligned}
  \dot{{a}}(t)  = &- z(t), \\
  \dot{z}(t)  = &  A{a}(t)+Bz(t)  + \sigma(a(t),z(t))\xi(t),
\end{aligned}
\right.
\label{eq: 1stmul}
\end{equation}
where $A=A_{1}$ and $B=B_{1}$ are given in \eqref{eq: A1B1}. This is a Langevin equation. But notice that in the multiplicative noise term, we allowed the diffusion coefficient to depend on both $a$ and $z.$

Again from the results of direct simulations,  we observed that the components of $a$ and $z$ are  independent. So we assume all these components are independent and its joint distribution function is written as,
\begin{equation}
\rho(a,z) \propto \prod_{i=1}^{n} \left({a}_{i}-\frac{\alpha_{i}}{\beta_{i}}\right)^{\alpha_{i}-1}\!\!\!\!\!e^{-\beta_{i}\left({a}_{i}+\alpha_{i}/\beta_{i}\right)}e^{-\gamma_{i}|z_{i}|}.
\label{eq: azpdf}
\end{equation}
Similarly, we assume $\sigma$ to be diagonal and work with the steady state solution of the FPE, which can be written as,
\begin{equation}
\frac{\partial D_{ii}\rho}{\partial z_{i}} = 
	\left( \sum_{j=1}^{n}A_{ij}a_{j} + \sum_{j=1,j\neq i}^{n} B_{ij}z_{j} - \frac{1}{\gamma_{i}^{2}}\frac{\partial W(a)}{\partial a_{i}} \right) \rho
	+ B_{ii}z_{i}\rho
	- \frac{1}{\gamma_{i}}\frac{\partial W(a)}{\partial a_{i}}|z_{i}|\rho.
\label{eq: Diirho}
\end{equation}

By integrating this equation, we have,

\begin{theorem}
In \eqref{eq: 1stmul}, suppose $B$ has non-positive diagonal entries. We have for $i=1,2,\cdots,n$,
\begin{equation} 
\begin{aligned}
D_{ii} = &- \frac{sgn(z_{i})}{\gamma_{i}} \left( 1-e^{\gamma_{i}|z_{i}|} \right)\left( \sum_{j=1}^{n}A_{ij}a_{j} + \sum_{j=1,j\neq i}^{n} B_{ij}z_{j} - \frac{2}{\gamma_{i}^{2}}\frac{\partial W(a)}{\partial a_{i}} \right)+ f_{1}e^{\gamma_{i}|z_{i}|}\\
	&+\frac{1}{\gamma_{i}^{2}}\frac{\partial W(a)}{\partial a_{i}}z_i + f_{2}e^{\gamma_{i}|z_{i}|} + B_{ii}\left( -\frac{1}{\gamma_{i}}|z_{i}| - \frac{1}{\gamma_{i}^{2}} +\frac{1}{\gamma_{i}^{2}}e^{\gamma_{i}|z_{i}|}\right)+f_{3}e^{\gamma_{i}|z_{i}|},
\end{aligned}
\label{eq: Dii}
\end{equation}
where 
\begin{equation}\begin{aligned}
f_{1} &= \frac{1}{\gamma_{i}}\left( \sum_{j=1}^{n}|A_{ij}||a_{j}| + \sum_{j=1,j\neq i}^{n} |B_{ij}||z_{j}| + \frac{2}{\gamma_{i}^{2}}\left(\beta_{i}+\frac{\alpha_{i}-1}{a_{i}+\mu_{i}}\right) \right) ,\\
f_{2} &=  \frac{1}{\gamma_{i}^{3}e}\left( \beta_{i} + \frac{\alpha_{i}-1}{a_{i}+\mu_{i}} \right),\\
f_{3} &= -B_{ii}\frac{1}{\gamma_{i}^{2}},
\end{aligned}\end{equation}
Further, $D_{ii}$ are positive and $\rho$ given in \eqref{eq: azpdf} is a steady state solution of the FPE,   if $\sigma$ is a diagonal matrix and $2\sigma_{ii}^{2}=D_{ii}$.
\end{theorem}

From Figure \ref{fig: mul1}, we see that the first order model \eqref{eq: 1stmul} is able to capture the statistics of both energy $a$ and $z$. Similar to the zeroth order model, the SDEs are stiff, and the implicit Taylor method \cite{Tian2001implicit} with stepsize $\Delta t=5\times 10^{-4} $ is used to generate a long time series to obtain the statistics. 
\begin{figure}[hbtp]
\begin{center}
\includegraphics[width=0.45\textwidth]{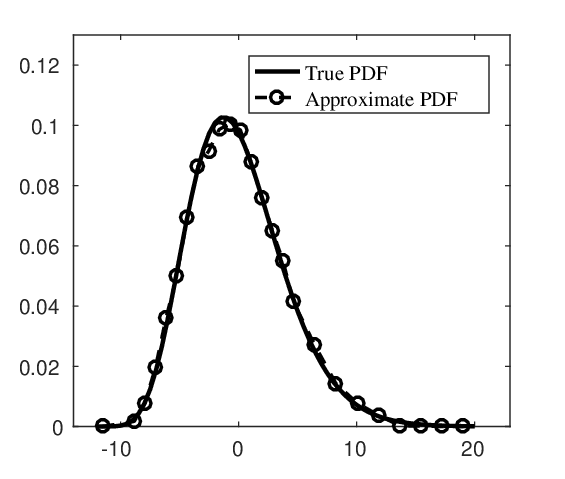}
\includegraphics[width=0.45\textwidth]{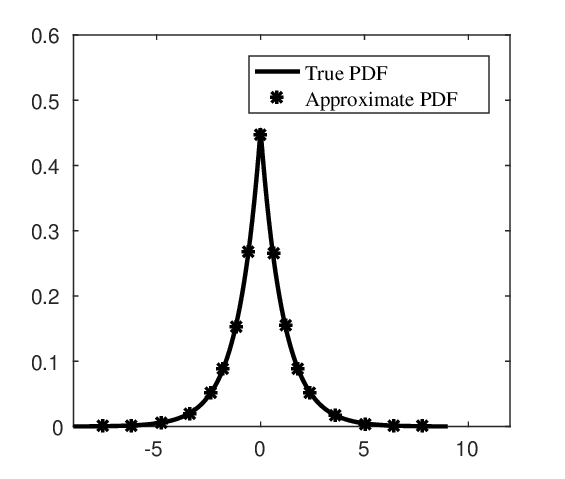}
\caption{The left figure shows the comparison of the true PDF of $a_1$ and its approximation obtained from solving the first order SDE model \eqref{eq: 1stmul} using the implicit Taylor method. The right panel shows the PDFs of $z_1$.}
\label{fig: mul1}
\end{center}
\end{figure}

\section{Summary and discussions}

This work is concerned with a coarse-grained energy model directly obtained from the full molecular dynamics model. The goal is to find a more efficient model so that the heat conduction process can be simulated with a model that is a lot cheaper than non-equilibrium molecular dynamics simulations. Our focus has been placed on the equilibrium statistics of such models, which conceptually, is often a good starting point to develop a stochastic model.  Unlike many of the coarse-grained molecular models, in which the coarse-grained velocity is expected to have Gaussian statistics, we found that the coarse-grained energy follows non-Gaussian statistics. We proposed to introduce multiplicative noise, within the Markovian embedding framework for the memory term, to ensure that the solution of the coarse-grained models has the correct equilibrium statistics.  

Although we only considered a one-dimensional model as the example, the framework is applicable to more general systems. In particular,  none of the theorems assumed the space dimensionality. The applications to nano-mechanical systems is currently underway. 

This work suggests a new paradigms for modeling nano-scale transport phenomena in general: Rather than relying on traditional deterministic models, e.g., the Fourier's Law and the Fick's Law, one may derive from first-principle a {\it stochastic and nonlocal} constitutive relation for the mass or heat current. 
The extension of the current effort to general diffusion processes will be in the scope of our future work. 

\bibliographystyle{unsrtnat}

\bibliography{Mori-Heat}
\end{document}